\documentclass[12pt]{amsart}
\usepackage{enumitem, verbatim}
\usepackage{color}
\usepackage{appendix}
\usepackage{amsmath, amssymb, amsthm}
\usepackage{mathrsfs}

\usepackage[margin=3 cm,heightrounded=true,centering]{geometry}
\usepackage{graphicx}
\usepackage{wrapfig}
\usepackage{cancel}
\newtheorem{theorem}{Theorem}[section]
\newtheorem{proposition}[theorem]{Proposition}
\newtheorem{lemma}[theorem]{Lemma}

\newtheorem{corollary}[theorem]{Corollary}
\newtheorem{definition}[theorem]{Definition}

\usepackage{hyperref}
\hypersetup{colorlinks=true, linkcolor=blue, citecolor = blue, urlcolor=blue, filecolor=magenta}

\newcommand{\td}{\text{d}}
\DeclareMathOperator{\ric}{Ric}

\DeclareMathOperator{\tr}{tr}

\DeclareMathOperator{\SU}{SU}

\DeclareMathOperator{\U}{U}

\DeclareMathOperator{\ad}{ad}
\DeclareMathOperator{\diag}{diag}

\DeclareMathOperator{\divergence}{div}

\title{Extreme 5-dimensional black holes with $\SU(2)$-symmetric horizons}

\author{Eric Bahuaud}
\address[Eric Bahuaud]{Department of Mathematics, Seattle University, Seattle, WA 98122, USA}
\email{bahuaude@seattleu.edu}

\author{Sharmila Gunasekaran}
\address[Sharmila Gunasekaran]{Department of Mathematics, IMAPP, Radboud University, 6500 GL Nijmegen, The Netherlands}
\email{s.gunasekaran@math.ru.nl}

\author{Hari K Kunduri}
\address[Hari K Kunduri]{Department of Mathematics and Statistics and Department of Physics and Astronomy, McMaster University, Hamilton ON, Canada L8S 4K1}
\email{kundurih@mcmaster.ca}

\author{Eric Woolgar}
\address[Eric Woolgar]{Department of Mathematical and Statistical Sciences and Theoretical Physics Institute, University of Alberta, Edmonton, Alberta, Canada T6G 2G1}
\email{ewoolgar@ualberta.ca}

\begin{document}

\begin{abstract}
We show that the near horizon geometry of 5-dimensional extreme (i.e., degenerate) stationary vacuum black holes, with or without cosmological constant, whose event horizons exhibit $\SU(2)$ symmetry must be that of a Berger sphere.
\end{abstract}

\maketitle

\section{Introduction}\label{section1}
\setcounter{equation}{0}

\noindent A standard anthropic argument that space cannot have more than three extended dimensions is that only in three dimensions will angular momentum and Newtonian gravitational attraction balance each other to create an effective potential well in which planets can stably orbit stars in bounded orbits, creating the conditions for life.\footnote
{In $n$-dimensional space, the effective Newtonian potential felt by a particle with angular momentum $L$ at distance $r$ from a mass $M>0$ that produces a Newtonian gravitational field is $V=-\frac{M}{r^{n-2}}+\frac{L^2}{r^2}$, which will not have a minimum if $n\ge 4$.}
The analysis continues to hold if one replaces Newtonian gravity with general relativity. The general relativistic argument rests on our knowledge of vacuum solutions of general relativity corresponding to isolated systems with $n$ extended spatial dimensions in an $(n+1)$-dimensional spacetime for $n\ge 4$. Indeed, this motivated the discovery of the Schwarzschild-Tangherlini metric \cite{Tangherlini}, describing a static (i.e., stationary and non-rotating) black hole in an $(n+1)$-dimensional, asymptotically flat spacetime. A class of higher dimensional analogues of the rotating Kerr black hole are also known and can be written down exactly in explicit coordinate form. These are called Myers-Perry black holes (\cite{MP}, see \cite{Myers} for a review; metrics of Myers-Perry type with cosmological constant, which we will sometimes call Myers-Perry-(A)dS metrics, were found in \cite{GLPP}).

Five-spacetime-dimensional Myers-Perry black holes have an event horizon with closed spatial cross-sections of $3$-sphere topology. Since the spacetime is stationary (rotating, but at a constant rate), the event horizon is also a Killing horizon, a fact which we will exploit. Some of these black holes, those which have equal angular momenta in two orthogonal 2-planes, have horizon cross-sections which exhibit $\SU(2)$ symmetry with a 4-dimensional isometry group and are \emph{Berger spheres} (sometimes in the physics literature called \emph{squashed spheres}). No stationary $5$-dimensional black hole with $\SU(2)$-symmetric horizon and only a $3$-dimensional maximal isometry group (the generic dimension of the maximal isometry group for left-invariant metrics on
 $\SU(2)$) is known. Then why not? Have we not looked hard enough perhaps? Or is it impossible to find an interpolating spacetime which has a standard asymptotic region with round conformal infinity?

Herein we show that the answer is simply that a local condition on the horizon suffices to rule out the case of a 3-dimensional maximal isometry group when the black hole is \emph{extreme}---i.e., when the horizon is a \emph{degenerate Killing horizon}. This local condition is encoded in the near horizon geometry equation, which is induced on the horizon by the Einstein equation. The argument is entirely confined to the horizon, and so it holds without regard to the spacetime asymptotics. If the horizon is the Lie group $\SU(2)$, the isometry group will contain $\SU(2)\times \U(1)$ and the near horizon geometry will be that of a Berger sphere. In addition to the extreme black hole case, our arguments also apply for non-extreme metrics if the so-called transverse null second fundamental form of the horizon is diagonal in a Milnor frame.
Recall that on any unimodular 3-dimensional Lie group (all compact Lie groups are unimodular, hence $\SU(2)$ is unimodular) there is an orthonormal basis $\left \{ E_i \right \}$ for left-invariant vector fields, called a \emph{Milnor frame}, such that $[E_i,E_{i+1}]=c_{i+2}E_{i+2}$ where $c_i$ are constants and $i$ is an integer mod $3$ \cite[Section 4]{Milnor} .

We need the notion of the \emph{near horizon geometry} of a degenerate Killing horizon. By this, we mean a solution of the equation
\begin{equation}
\label{eq1.1}
\ric_X^m:= \ric +\frac12 \pounds_X\gamma -\frac{1}{m} X\otimes X =\Lambda \gamma +\kappa B,
\end{equation}
for a Riemannian metric $\gamma$, a covector field $X$, $m=2$, $\kappa=0$, and $\Lambda$ an arbitrary constant. Much or our analysis will hold for arbitrary (positive) values of $m$, and sometimes for arbitrary $\kappa$. These cases have applications; e.g., see \cite{GKW} for an application of the $m=1$ case to cosmology. When $\kappa\neq 0$, then $B$ is a symmetric $(0,2)$-tensor, and even for the applications of interest in this paper the arbitrariness of $m$ will be useful. Note that in equation \eqref{eq1.1}, in a common abuse of notation, $X$ denotes both a 1-form and the vector field related to it by metric duality (i.e., by raising an index with $\gamma^{-1}$).

When $m=2$, this equation arises from the Einstein equations on a Killing horizon in a stationary spacetime. See \cite[Appendix]{HIW}, and set $r=0$ in equation (82) of that reference to obtain our equation \eqref{eq1.1}. For a review in the case of a degenerate Killing horizon (so that $\kappa=0$), see \cite{KL2}. As alluded to above, in \cite{HIW} the quantity $B$ has the interpretation of being the null second fundamental form of the Killing horizon in the null transverse spacetime direction ${\ell}$; i.e., $B=\pounds_{\ell}\gamma$. Then $\ell$ is a null vector transverse to the Killing horizon, which is an embedded null hypersurface in spacetime and in which $(M,\gamma)$ is a cross-section. In the spacetime context, $\Lambda$ has the interpretation of the spacetime cosmological constant divided by the dimension of $M$ which, in turn, is the dimension of spacetime minus $2$. Finally, the constant $\kappa$ is called the \emph{surface gravity} of the Killing horizon.\footnote
{The terminology (and the constancy of $\kappa$) derives from the $m=2$ case which is our main interest here. For, say, the application to cosmology with $m=1$ in \cite{GKW}, the terminology is less well suited.}
What we call the \emph{degenerate} or \emph{extreme} case is the case of $\kappa=0$.

\begin{theorem}\label{theorem1.1}
Let $(M,\gamma, X)$ be a solution of \eqref{eq1.1} with $\kappa=0$ and $m>0$. If $M$ is the Lie group $\SU(2)$ and $\gamma$ is a left-invariant metric on $M$, then the isometry group of $(M,\gamma)$ contains $\SU(2)\times \U(1)$ and the metric is that of a Berger sphere.
\end{theorem}

A version of this result appeared in \cite[Corollary 4.1]{KL2}, relying on a result in \cite{KL1} that \emph{assumed} that $X$ is a Killing vector field. Because we only work with near horizon geometries of dimension $3$, we are able to prove that $X$ is a Killing vector field, eliminating the need to assume it in this case. Proofs of this fact were given in work by other authors (\cite{Lim}, \cite{CLZ}), who did not notice the application to black holes. We are also able to show how existence of this Killing field eliminates a certain non-Berger family of metrics from consideration, verifying an implicit assumption of \cite{KL1}. Our arguments are purely Lie-theoretic, in the spirit of \cite{Lim}.

In five-dimensional gauged supergravity, whose bosonic sector can be thought of as general relativity with a negative cosmological constant and coupled to a Maxwell field, it has been recently proved that supersymmetric black hole solutions (i.e., those possessing Killing spinors) with $\SU(2)$ symmetry must have an enhanced $\SU(2)\times \U(1)$ symmetry with associated horizon cross sections being Berger spheres \cite[Theorem 1]{Lucietti:2021}. In fact, in $5$-dimensional (ungauged) supergravity, which has no cosmological constant, the near horizon geometry of any supersymmetric black hole with $S^3$ topology must be that of a Berger sphere \cite[Theorem 1]{Reall}.

A question which remains open is whether the horizon geometry of a nondegenerate $\SU(2)$ Killing horizon is subject to similar restrictions, which is what is seen in known $\SU(2)$-invariant non-extreme Myers-Perry type metrics. This question cannot be fully answered by studying equation \eqref{eq1.1} alone, since for non-vanishing $\kappa$ one can prescribe a left-invariant $B$ that does not respect the full isometry group. But one can ask for conditions on $B$ such that Theorem \ref{theorem1.1} applies without the assumption that $\kappa= 0$.

\begin{theorem}\label{theorem1.2}
For any $m> 0$ and any $\Lambda$, there exist left-invariant metrics $\gamma$ on $\SU(2)$ that solve equation \eqref{eq1.1} with nonvanishing left-invariant $X$ and nonvanishing $\kappa$ such that $\pounds_X\gamma=0$ and $B$ is left-invariant and diagonal in the Milnor frame. A sufficient, but not a necessary, condition for the conclusions of Theorem \ref{theorem1.1} to apply is that $B$ has a repeated eigenvalue belonging to the eigenspace orthogonal to $X$.
\end{theorem}

Section 2.1 contains a brief discussion of left-invariant metrics on  $\SU(2)$ whose Ricci tensor has a repeated eigenvalue. Equation \eqref{eq1.1} with $\kappa=0$ was analyzed on 3-dimensional Lie group manifolds by Lim \cite{Lim}, whose work generalizes the seminal paper of Milnor \cite{Milnor} on Ricci curvature of 3-dimensional Lie groups. Almost no new results beyond those already present in these papers are needed to prove Theorem \ref{theorem1.1}. Since Lim omitted some intermediate results in her exposition, we revisit these results in Section 2.2 and provide brief derivations to fill gaps. Similar results can be found in \cite{CLZ}, often with distinct arguments. We prove Theorem \ref{theorem1.1} in Section 2.3 and follow it with a brief discussion in Section 2.4. Section \ref{section3} contains the proof of Theorem \ref{theorem1.2}. A short appendix gives a brief discussion of $\SU(2)$-invariant Myers-Perry-(A)dS $5$-dimensional spacetimes and the corresponding near horizon geometries.

\subsection*{Acknowledgements} The research of EB was partially supported by a Simons Foundation Grant (\#426628, E Bahuaud). The research of HK was supported by NSERC grant RGPIN-2018-04887.
The research of EW was supported by NSERC grant RGPIN--2022--03440. EW thanks Will Wylie for a discussion of reference \cite{Lim} and for bringing reference \cite{CLZ} to his attention. HK thanks McKenzie Wang for useful discussions.

\subsection*{Data statement} Data sharing is not applicable to this article as no datasets were generated or analyzed during the current study.

\subsection*{Conflict of interest statement} The authors have no conflicts of interest.

\section{3-dimensional Lie groups and the main theorem}\label{section2}
\setcounter{equation}{0}

\subsection{Ricci eigenvalues}
Up to homothety, we may write an arbitrary left-invariant metric on $\SU(2)$ as
\begin{equation}
\label{eq2.1}
\gamma = \varepsilon^2 \sigma^1 \otimes \sigma^1 + \beta^2 \sigma^2 \otimes \sigma^2 + \sigma^3 \otimes \sigma^3
\end{equation}
where $\varepsilon,\beta\in (0,1]$ are constants and the $\sigma^i$ are left-invariant 1-forms. The left-invariant vector fields $\sigma_i$ belonging to the dual basis ($\sigma^i(\sigma_j) = \delta^i_{~j}$) obey the relations $\left [ \sigma_i,\sigma_{i+1} \right ] =2\sigma_{i+2}$ where the indices are integers mod $3$.
The Ricci tensor is diagonal in this basis. In particular, the Ricci endomorphism of this metric has eigenvalues
\begin{equation}
\label{eq2.2}
\begin{split}
\rho_1 = &\, 2\frac{\left ( \varepsilon^4 - (1-\beta^2 )^2\right ) }{\varepsilon^2\beta^2}, \\
\rho_2 = &\, 2\frac{\left ( \beta^4 - (1-\varepsilon^2 )^2\right )}{\varepsilon^2\beta^2}, \\
\rho_3 = &\, 2\frac{\left ( 1 - (\varepsilon^2 - \beta^2)^2\right )}{\varepsilon^2\beta^2}.
\end{split}
\end{equation}
Then $\rho_3 \ge \max \{\rho_1,\rho_2\}$; i.e., for a left-invariant metric on $\SU(2)$ the eigenspace of the largest Ricci eigenvalue corresponds to the direction with the largest metric coefficient in \eqref{eq2.1}. The unit round metric on ${\mathbb S}^3$ is recovered when $\varepsilon=\beta=1$, and in this case $\ric =2 \gamma$.

\begin{definition}\label{definition2.1}
The \emph{Berger sphere metrics} are given by equation \eqref{eq2.1} when two metric coefficients are equal (i.e., either $\beta =1$ or $\varepsilon =1$ or $\varepsilon = \beta$; if all three are equal, the metric is the bi-invariant round $3$-sphere metric with $6$-dimensional isometry group).
\end{definition}

Berger spheres have isometry group that contains $\SU(2)\times \U(1)$. From \cite[Theorem 3.8]{HL} (or \cite[Theorem 2.10]{Shin}), any left-invariant metric on $\SU(2)$ that is not a Berger sphere has only a $3$-dimensional maximal isometry group. For the Berger spheres, two of the eigenvalues of the Ricci endomorphism are also equal. The next two lemmata provide a converse.

\begin{lemma}\label{lemma2.2}
Consider a left-invariant metric on $\SU(2)$. If two eigenvalues of the Ricci endomorphism are equal, either (i) the metric is that of a Berger sphere or (ii) $\beta^2=1-\varepsilon^2$ in \eqref{eq2.1} and then the Ricci endomorphism has eigenvalues $(0,0,8)$.
\end{lemma}

\begin{proof}
Equality of two eigenvalues leads to the following conditions on $\beta$ and $\varepsilon$:
\begin{eqnarray}
\label{eq2.3} \rho_2=\rho_3&\, \implies &\, \left ( 1-\beta^2\right ) \left ( 1+\beta^2-\varepsilon^2\right ) =0,\\
\label{eq2.4}\rho_1=\rho_3&\, \implies &\, \left ( 1-\varepsilon^2\right ) \left ( 1+\varepsilon^2-\beta^2\right )=0,\\
\label{eq2.5} \rho_1=\rho_2&\, \implies &\,  \left ( \varepsilon^2-\beta^2\right ) \left ( 1-\varepsilon^2 -\beta^2\right ) =0.
\end{eqnarray}
If the first factor in any of \eqref{eq2.3}--\eqref{eq2.5} vanishes, two metric coefficients in \eqref{eq2.1} will be equal so the metric is a Berger sphere. The Ricci endomorphism then can have signature either $(+,+,+)$, $(0,0,+)$, or $(-,-,+)$ depending on the value of the remaining free parameter $\varepsilon$ (or $\beta$).

Since $\varepsilon^2,\beta^2\in (0,1]$, the only one of these three equations in which the second factor can vanish is \eqref{eq2.5}, and then $\varepsilon^2+\beta^2=1$. From \eqref{eq2.2}, the eigenvalues of the Ricci endomorphism are then $(0,0,8)$.
\end{proof}

Next, we will need the following lemma, which singles out the Berger sphere $\beta^2=\varepsilon^2=\frac12$ amongst all metrics of the form \eqref{eq2.1} with $\beta^2=1-\varepsilon^2$.

\begin{lemma}\label{lemma2.3}
Let $\gamma$ be a left-invariant metric on $\SU(2)$  of the form \eqref{eq2.1} with $\beta^2=1-\varepsilon^2$. If $E_3$ is a Killing vector field then $\beta^2=\varepsilon^2=\frac12$ and the metric is a Berger sphere.
\end{lemma}

\begin{proof}
If $E_3 = \sigma_3$ is a Killing vector field, then evaluating  $\pounds_{\sigma_3} \gamma=0$ on the pair $(\sigma_1, \sigma_2)$ yields
\begin{equation}
\label{eq2.6}
\pounds_{\sigma_3} \gamma ( \sigma_1, \sigma_2 ) = \sigma_3 \gamma(\sigma_1,\sigma_2) - \gamma( [\sigma_3,\sigma_1],\sigma_2 ) - \gamma( \sigma_1, [\sigma_3, \sigma_2] ) = -2 (\beta^2 -  \varepsilon^2 ) = 0,
\end{equation}
from which we conclude that $\varepsilon = \beta$. Since $\beta^2=1-\varepsilon^2$ by assumption, then $\varepsilon^2=\beta^2=\frac12$ and \eqref{eq2.1} then takes the form of a Berger sphere metric.
\end{proof}

It is interesting that the left-invariant $\SU(2)$ metrics \eqref{eq2.1} with $\beta^2=1-\varepsilon^2$ have the same Ricci eigenvalues. Since the manifold is $3$-dimensional, they therefore have the same Riemann curvature tensor, yet they are not all isometric. One of them is the metric with $\varepsilon^2=\beta^2=\frac12$, which is a Berger sphere. For this metric, the eigenvector $E_3$ spanning the eigenspace of largest Ricci eigenvalue is a Killing vector, but this is not the case for any other member of the $\beta^2=1-\varepsilon^2$ family of metrics. This proves that the Berger sphere metric with $\varepsilon^2=\beta^2=\frac12$ is not isometric to any other member of the family. The other family members are not Berger spheres, and so have only a $3$-dimensional maximal isometry group \cite[Theorem 3.8]{HL} (or \cite[Theorem 2.10]{Shin}). For an interesting discussion of metrics with the same curvatures that are not necessarily isometric, see \cite[pp 213--216]{Berger} and \cite[Theorem 6]{Kulkarni}.

\subsection{Quasi-Einstein Lie groups} In this section we recall some of the background we will need from Lim \cite{Lim} (and Milnor \cite{Milnor}) and fill in some details. We also extend relevant parts of the analysis to solutions of equation \eqref{eq1.1} with non-vanishing $B$, as needed for Theorem \ref{theorem1.2}. Our goal is to find a condition on $\kappa B$ under which the vector field $X$ in equation \eqref{eq1.1} is a Killing vector field, so that the Lie derivative term vanishes and \eqref{eq1.1} simplifies. Indeed, Lim observed that a sufficient condition for this purpose on $3$-dimensional unimodular Lie groups is that $\kappa$ vanishes. We will exploit the fact that the Lie group $\SU(2)$ is compact, which implies that $\SU(2)$ is unimodular. This in turn implies that for $X$ in the Lie algebra, the linear map $\mathrm{ad}_X=[X,\cdot ]$ from the Lie algebra to itself is tracefree.

We begin with a technical lemma.

\begin{lemma}\label{lemma2.4}
Let $\gamma$ be a left-invariant metric on a unimodular Lie group with Levi-Civita connection $\nabla$ and let $\{ E_i \}$ be an orthonormal basis of left-invariant vector fields. If $T$ is a left-invariant symmetric $(0,2)$-tensor, and  $X$ is a left-invariant vector field, then
\begin{equation}
\label{eq2.7}
\tr_{\gamma} \left ( T\circ \ad_X\right ) =\left \langle X, \divergence_{\gamma} T\right \rangle,
\end{equation}
where $\ad_X(Y):=[X,Y]$ for any left-invariant vector field $Y$, and $\divergence_{\gamma} T:=\nabla\cdot T$.
\end{lemma}

\begin{proof}
The proof is a straightforward calculation:
\begin{equation}
\label{eq2.8}
\begin{split}
\tr_{\gamma}\left ( T\circ \ad_X\right ) =&\, \sum_{i,j} \gamma^{ij} T\left ( [X,E_i],E_j\right ) = \sum_i T\left ( [X,E_i],E_i\right )\text{ by orthonormality,}\\
=&\, \sum_i T\left ( \nabla_X E_i,E_i\right ) -\sum_i T\left ( \nabla_{E_i}X,E_i\right )\text{ since }\nabla\text{ is torsion-free,}\\
=&\, \frac12 \sum_i \left [  T\left ( \nabla_X E_i,E_i\right ) + T\left ( E_i,\nabla_X E_i\right ) \right ] -\sum_i T\left ( \nabla_{E_i}X,E_i\right )\\
=&\, \frac12 \sum_i X\left ( T(E_i,E_i)\right ) -\frac12 \sum_i \left ( \nabla_X T\right ) (E_i,E_i)\\
&\, -\sum_i E_i \left ( T(X,E_i)\right ) +\sum_i \left ( \nabla_{E_i}T\right ) (X,E_i) +\sum_i T\left ( X, \nabla_{E_i}E_i\right ),
\end{split}
\end{equation}
where the last line is just the Leibniz rule. But since $X$, $T$, and each of the $E_i$ are left-invariant, then
\begin{itemize}
\item [(i)] $X\left ( T(E_i,E_i)\right )=0$ and
\item [(ii)] $ E_i \left ( T(X,E_i)\right )=0$.
\end{itemize}
Moreover,
\begin{equation}
\label{eq2.9}
\sum_i \left ( \nabla_X T\right ) (E_i,E_i) = \tr_{\gamma} \nabla_X T =\nabla_X \tr_{\gamma} T,
\end{equation}
using compatibility of $\nabla$ and $\gamma$, and by left invariance we then obtain
\begin{itemize}
\item [(iii)] $\sum_i \left ( \nabla_X T\right ) (E_i,E_i) = \nabla_X \tr_\gamma T=0$.
\end{itemize}
Then \eqref{eq2.8} reduces to
\begin{equation}
\label{eq2.10}
\tr_{\gamma}\left ( T\circ \ad_X\right ) =\sum_i \left ( \nabla_{E_i}T\right ) (X,E_i) +\sum_i T\left ( X, \nabla_{E_i}E_i\right ).
\end{equation}

Finally, our argument to this point allows completely general left-invariant $T$.  Specializing \eqref{eq2.10} to $T = \gamma$ and using that the group is unimodular, which implies that $\tr \ad_X =0$ for any left-invariant vector field $X$, we get that
\begin{equation}
\label{eq2.11}
0=\sum_i \gamma\left ( X, \nabla_{E_i}E_i\right ) = \gamma( X, \sum_i \nabla_{E_i} E_i).
\end{equation}
So we conclude $\sum_i \nabla_{E_i} E_i = 0$. Plugging this back into \eqref{eq2.10} yields \eqref{eq2.7}.
\end{proof}

The following corollary seems not to have been  explicitly stated in \cite{Lim}, but played an important role in the analysis therein.

\begin{corollary}\label{corollary2.5}
$\tr_{\gamma}\left ( \ric\circ \ad_X\right ) =0$.
\end{corollary}

\begin{proof}
By \eqref{eq2.7} with $T=\ric$ and using that $\divergence \ric = \frac12 \nabla R$, we have that
\begin{equation}
\label{eq2.12}
\begin{split}
\tr_{\gamma}\left ( \ric\circ \ad_X\right ) =&\, \frac12 \left \langle X, \nabla R \right \rangle =\frac12 X(R)=0,
\end{split}
\end{equation}
since left invariance of $\gamma$ implies that the scalar curvature $R$ is constant.
\end{proof}

\begin{corollary}\label{corollary2.6}
If $G$ is a unimodular  Lie group with a left-invariant metric $\gamma$ that obeys \eqref{eq1.1} with left-invariant fields $B$ and $X$ and constants $\kappa$, $\Lambda$, and $m$, then
\begin{equation}
\label{eq2.13}
\tr_{\gamma}\left \{\left ( \frac12 \pounds_X \gamma -\frac{1}{m}X\otimes X\right ) \circ \ad_X\right \} = \kappa \left \langle X, \divergence B \right \rangle .
\end{equation}
\end{corollary}

\begin{proof}
We compose each of $\ric$, $\gamma$, and $B$ with the adjoint map and take the trace of each such composition. By Lemma \ref{lemma2.4} and Corollary \ref{corollary2.5}, $\tr_{\gamma}(\ric\circ\ad_X)=0$ and $\tr_{\gamma}(\gamma\circ\ad_X)=0$, while $\tr_{\gamma} (B\circ\ad_X)=\left \langle X, \divergence B \right \rangle$. But traces are linear so we can combine these three traces to obtain
\begin{equation}
\label{eq2.14}
\begin{split}
\tr_{\gamma}\left \{ -\ric\circ \ad_X + \Lambda \gamma\circ \ad_X +\kappa B\circ\ad_X \right \}
=&\,  \tr_{\gamma}\left \{\kappa B\circ\ad_X \right \} = \kappa \left \langle X, \divergence B \right \rangle .
\end{split}
\end{equation}
Using \eqref{eq1.1} to simplify the left-hand side of \eqref{eq2.14}, we obtain \eqref{eq2.13}.
\end{proof}

We are now able to state the key lemma that we seek.

\begin{proposition}\label{proposition2.7}
If $G$ is a unimodular Lie group with a left-invariant metric $\gamma$ that obeys \eqref{eq1.1} with left-invariant fields $B$ and $X$ and constant $\kappa$, and if
\begin{equation}
\label{eq2.15}
\kappa \left \langle X, \divergence B \right \rangle =0,
\end{equation}
then $X$ is a Killing vector field; i.e., $\pounds_X \gamma =0$.
\end{proposition}

\begin{proof}
Define
\begin{equation}
\label{eq2.16}
\frac12 \pounds_X \gamma -\frac1m X\otimes X = q,
\end{equation}
where $\gamma$ and $X$ are left-invariant fields on a unimodular Lie group. Then $\ric$ is left-invariant, and since $B$ is assumed to be left-invariant as well, then by \eqref{eq1.1} we have that $q$ is also left-invariant. Then we must show that
\begin{equation}
\label{eq2.17}
\tr_{\gamma} \left ( q\circ \ad_X\right )=0.
\end{equation}
But by Corollary \ref{corollary2.6}, equation \eqref{eq2.17} will hold whenever whenever assumption \eqref{eq2.15} applies, and then \cite[Lemma 2.4]{Lim} implies that $\pounds_X \gamma =0$.

For a second proof when $G$ is compact and $\kappa=0$, since $X$ is left-invariant it is divergence-free. Then \cite[Theorem 1.1]{BGKW2} implies that $\pounds_X\gamma=0$.

\end{proof}

We note that \eqref{eq2.15} holds whenever $B$ is a Codazzi tensor on a Lie group; i.e., whenever $\nabla_iB_{jk}=\nabla_jB_{ik}$ [proof: contracting, we get $\divergence_{\gamma}B=\nabla_j \tr_{\gamma} B=0$, since $\tr_{\gamma}B$ is constant when $B$ and $\gamma$ are left-invariant]. A special case of this is the case of umbilic $B$; i.e., $B = \mu \gamma$ for $\mu\in {\mathbb R}$.

Of course, whenever $\gamma$ is the near-horizon metric of a degenerate Killing horizon then $\kappa=0$ and so then $X$ is Killing, as observed in \cite{Lim}.

For compact Lie groups (such as $\SU(2)$), another way to write the condition \eqref{eq2.15} is as follows. From the Leibniz rule, we have that
\begin{equation}
\label{eq2.18}
\left \langle X, \divergence B \right \rangle =\divergence \left ( B(X)\right ) -B^{ij}\nabla_iX_j,
\end{equation}
raising indices with $\gamma^{-1}$. Since $B$ and $\gamma$ are left-invariant (from which it follows that the Levi-Civita connection $\nabla$ is left-invariant), then the first term on the left is the divergence of a left-invariant vector field, which is always zero on a compact Lie group. Then \eqref{eq2.15} becomes
\begin{equation}
\label{eq2.19}
0= \left \langle B, \pounds_X\gamma\right \rangle .
\end{equation}

Finally, we show that the assumption that $X$ is left-invariant is redundant in our setting, since it follows from the left invariance of $q$ and $\gamma$. This is the essential content of \cite[Lemma 2.2]{Lim}, whose proof may contain a minor error (as it considers the map $\ad_X$ before it is known that $X$ belongs to the Lie algebra). As our interest here is $\SU(2)$, we will limit our proof to compact groups. See also \cite[Theorem 2.1]{CLZ}, whose proof assumed the equation \eqref{eq1.1} (our proof does not).

\begin{lemma}\label{lemma2.8}
Let $G$ be a compact 3-dimensional Lie group with left-invariant metric $\gamma$. Let $q$ and $X$ be as in equation \eqref{eq2.16}. If $q$ is left-invariant then so is $X$.
\end{lemma}

\begin{proof} We suspend the summation convention in this proof and write all summations explicitly. Since $G$ is compact, it is unimodular and so admits a Milnor frame $\left \{ E_i \right \}$. Note that in such a frame $\langle[E_i,E_j],E_i\rangle=0$ for any $i,j$. Let $X = \sum_i f_i E_i$. We compute, using the Leibniz rule for the Lie derivative $\pounds_X (\gamma (Y,Z) ) = \pounds_X \gamma(Y,Z) + \gamma( \pounds_X Y, Z) + \gamma( Y, \pounds_X Z)$, at an arbitrary point in $G$.
\begin{align}
\label{eq2.20}
\begin{aligned}
\frac{1}{2} \pounds_X \gamma(E_i,E_i) =&\, \frac{1}{2} X ( \gamma(E_i,E_i) ) - \gamma( \pounds_X E_i, E_i ) = -\gamma ( [X,E_i], E_i ) \\
=&\,  -\gamma \left( \left[ \sum_j f_j E_j ,E_i\right], E_i \right) \\
=&\,  -\gamma \left(  \sum_j f_j \left[E_j ,E_i\right] - (E_i f_j) X_j, E_i \right) \\
=&\,  E_i f_i - \sum_j f_j \gamma \left(  \left[E_j ,E_i\right], E_i \right). \\
\end{aligned}
\end{align}
Then
\begin{align}
\label{eq2.21}
\begin{aligned}
\frac{1}{2} \pounds_X \gamma(E_i,E_i) - \frac{1}{m} \left \langle X,E_i\right \rangle^2 = E_i f_i - \frac{1}{m} f_i^2.
\end{aligned}
\end{align}

The remainder of the proof is exactly as given in \cite[Lemma 2.2]{Lim}. Since $G$ is compact, let $f_i$ achieve its maximum at $r\in G$ and its minimum at $s\in G$. Define $\mu_i:= q( E_i, E_i )$. Applying \eqref{eq2.21} and the hypothesis on $q$ we find that
\begin{align}
\label{eq2.22}
\mu_i = -\frac{1}{m} f_i^2(r) = -\frac{1}{m} f_i^2(s),
\end{align}
so that $f_i(r)^2 = f_i(s)^2 = - m \mu_i$. Let $c(t)$ be an integral curve of $E_i$. This integral curve exists for all $t$. In view of \eqref{eq2.21}, we see that along this curve $f_i(t)$ is a global (in ${\mathbb R}$) solution of
\begin{align}
\label{eq2.23}
f_i'(t) - \frac{1}{m} f_i^2 = \mu_i.
\end{align}
The nonconstant solutions of this equation are given by the hyperbolic tangent function with infinite domain of definition so it cannot achieve a maximum or minimum, which is a contradiction (see \cite[p 6]{Lim} for details). This leaves the constant solutions (and so the maximum and minimum of $f_i$ are equal). Since the $f_i$ are constant, $X=\sum_i f_iE_i$ is left-invariant.
\end{proof}

\subsection{Proof of Theorem 1.1}
Assume that $\kappa=0$ in \eqref{eq1.1}. As $\SU(2)$ is a compact $3$-dimensional Lie group, it is unimodular.  By \cite[Theorem 4.3]{Milnor}, there exists an orthonormal Milnor frame $\left \{ E_i \right \}$ of left-invariant vector fields that diagonalizes the Ricci tensor.

If $X$ vanishes, then $\ric = \Lambda \gamma$ and all three eigenvalues of the Ricci endomorphism are equal. Since $M = \SU(2)$, we conclude that $\Lambda>0$ and  $\gamma$ is a round sphere with $6$-dimensional isometry group, verifying the claim. So assume otherwise.

With $\kappa=0$, since $g$ is left-invariant and $\ric$ is natural, $\ric-\Lambda g$ is left-invariant. Then by Lemma \ref{lemma2.8}, so is $q$. Then
by Proposition \ref{proposition2.7}, $X$ is a Killing vector field. Since $\kappa=0$ as well, then \eqref{eq1.1} becomes
\begin{equation}
\label{eq2.24}
\ric=\Lambda \gamma +\frac{1}{m}X\otimes X .
\end{equation}
Since the Ricci tensor is diagonal in the Milnor frame, the off-diagonal components of \eqref{eq2.24} reduce to the conditions
\begin{equation}
\label{eq2.25}
X_iX_j=0,\quad i\neq j,
\end{equation}
and so at most one Milnor frame component of $X$ is nonzero, say $X=\alpha E_3$. But then, returning to \eqref{eq2.24}, the `1,1' and `2,2' components reduce to $\ric(E_1,E_1)=\Lambda=\ric(E_2,E_2)$, and so $\ric$ has a repeated eigenvalue $\Lambda$.
Since $\SU(2)$ is compact and $X$ is divergence-free, we know from \cite[Corollary 1.4]{BGKW2} that the Ricci endomorphism can have no more than two distinct eigenvalues. If the repeated eigenvalue is $\Lambda\neq 0$, by Lemma \ref{lemma2.2} the metric is a Berger sphere metric. If $\Lambda=0$, the metric may be a Berger sphere or it may belong to the family in Lemma \ref{lemma2.2} described by equation \eqref{eq2.1} with $\beta^2=1-\varepsilon^2$. However, $X$ is also Killing and is proportional to $E_3$, so the latter case is ruled out by Lemma \ref{lemma2.3}.

This completes the proof of Theorem 1.1.

\subsection{Discussion: Myers-Perry black holes}
By inserting $\beta=\varepsilon$ in the formula for the Ricci eigenvalues \eqref{eq2.2}, we have that Berger spheres obey $\ric= \diag \left ( \rho_1,\rho_2,\rho_3 \right ) =\frac{2}{\varepsilon^4} \diag \left (2\varepsilon^2-1, 2\varepsilon^2-1, 1\right )$ in the Milnor frame. Thus the Ricci tensor of a Berger sphere has signature $(+,+,+)$, $(0,0,+)$, or $(-,-,+)$. By equation \eqref{eq1.1} with $\pounds_X\gamma =0$ and writing that $X=k\sigma^3$ since $X$ can have only one nonzero component, we have
\begin{equation}
\label{eq2.26}
\begin{split}
\rho_1=&\,  \rho_2 = \frac{2\left ( 2\varepsilon^2-1\right )}{\varepsilon^4} = \Lambda ,\\
\rho_3=&\, \frac{2}{\varepsilon^4} = \Lambda + \frac{k^2}{m} ,
\end{split}
\end{equation}
Note that $\varepsilon\le 1$ implies that $\Lambda\le 2$.

By comparison to expressions \eqref{eqA.5} and \eqref{eqA.6} of the Appendix, we are able to express these quantities in terms of the physical parameters appearing in Myers-Perry black holes. In particular, this application requires that we set $m=2$ and then we see that the relations between $\varepsilon$, $k$, $\Lambda$, the Ricci eigenvalues $\rho_i$, and the black hole parameter $r_+$ introduced in the Appendix are
\begin{equation}
\label{eq2.27}
\begin{split}
\varepsilon=&\, r_+,\\
k^2=&\, 8\sqrt{1-\frac{\Lambda}{2}}\left ( 1 + \sqrt{1-\frac{\Lambda}{2}}\right )=\frac{8}{r_+^4}\left ( 1-r_+^2\right ), \\
\rho_1=&\, \rho_2=\Lambda = \frac{4}{r_+^2}-\frac{2}{r_+^4},\\
\rho_3=&\, \frac{2}{r_+^4} .
\end{split}
\end{equation}
The signature of the Ricci endomorphism is $(0,0,+)$ at $r_+=\frac{1}{\sqrt{2}}$, $(-,-,+)$ if $r_+< \frac{1}{\sqrt{2}}$, and $(+,+,+)$ if $r_+>\frac{1}{\sqrt{2}}$.

\section{Theorem 1.2} \label{section3}
\setcounter{equation}{0}

\subsection{Milnor frames and the divergence of \it{B}} A defining property of a Milnor frame for a unimodular 3-dimensional Lie group is that the Lie bracket of left-invariant basis vector fields takes the form $[E_1, E_2] =c_3 E_3$ and cyclic permutations; i.e., the structure constants $c^i_{jk}$ for the algebra obey
\begin{equation}
\label{eq3.1}
c_1:=c^1_{23}=-c^1_{32},\quad c_2:=c^2_{31}=-c^2_{13},\quad c_3:= c^3_{12}=-c^3_{21},\quad c^i_{jk}=0\text{ otherwise.}
\end{equation}
Using this fact, we now obtain an expression for the divergence of a left-invariant $(0,2)$-tensor (such as $B$) in a Milnor frame.

\begin{lemma}\label{lemma3.1}
On a unimodular Lie group with orthonormal Milnor frame $\{ E_i \}$, the divergence of a left-invariant symmetric $(0,2)$-tensor $B$ obeys
\begin{equation}
\label{eq3.2}
\divergence B=\sum\limits_{j,k}B^{jk} \nabla_{E_j}E_k .
\end{equation}
Indeed,
\begin{equation}
\label{eq3.3}
\begin{split}
\divergence B = &\, \left ( c_3-c_2\right ) B^{23} E_1 +\left ( c_1-c_3\right ) B^{31} E_2 +\left ( c_2-c_1\right ) B^{12} E_3,\\
\implies \left \langle X, \divergence B \right \rangle = &\, \left ( c_3-c_2\right ) B^{23} X_1 +\left ( c_1-c_3\right ) B^{31} X_2 +\left ( c_2-c_1\right ) B^{12} X_3,
\end{split}
\end{equation}
for any vector field $X$.
In particular, if $B$ is diagonal in the Milnor frame then $\divergence B=0$.
\end{lemma}

\begin{proof}
\begin{equation}
\label{eq3.4}
\begin{split}
\left \langle E_i,\divergence B\right \rangle =&\, \left \langle E_i,\sum\limits_j \left \langle \nabla_{E_j} B ,E_j\right \rangle \right \rangle \\
=&\, \left \langle E_i,\sum\limits_{j,k,q} \left \langle \nabla_{E_j} \left ( B^{kq}E_k\otimes E_q\right ) ,E_j\right \rangle \right \rangle \\
=&\,\sum\limits_{j,k,q} B^{kq} \left [ \left \langle E_i ,E_q\right \rangle \left \langle \nabla_{E_j}E_k,E_j\right \rangle + \left \langle E_k ,E_j\right \rangle \left \langle E_i, \nabla_{E_j} E_q\right \rangle\right ] \\
\end{split}
\end{equation}
Now $\left \langle \nabla_{E_j}E_k,E_j\right \rangle= \frac12 E_k \left ( \vert E_i\vert^2 \right )+\left \langle [E_j,E_k],E_j\right \rangle$. The first term vanishes because Milnor frames are orthonormal. The second term vanishes because, for a Milnor frame, $[E_j,E_k]$ is orthogonal to $E_j$. In the second term on the right in \eqref{eq3.4}, we also use the orthonormality to write that $\left \langle E_k ,E_j\right \rangle=\delta_{jk}$. Thus, \eqref{eq3.4} reduces to \eqref{eq3.2}.

Now we compute the elements of $\divergence B$ explicitly. Consider the $E_1$-component of $\nabla_{E_i}E_i$ (not summed on $i$). Obviously,
\begin{equation}
\label{eq3.5}
\left \langle E_1, \nabla_{E_1}E_1\right \rangle = \frac12 E_1\left ( \left \langle E_1, E_1\right \rangle\ \right ) =E_1(1)=0,
\end{equation}
since the Milnor frame is orthonormal. Next,
\begin{equation}
\label{eq3.6}
\begin{split}
\left \langle E_1, \nabla_{E_2}E_2\right \rangle = &\, E_2\left ( \left \langle E_1, E_2\right \rangle \right )-  \left \langle \nabla_{E_2}E_1, E_2\right \rangle \\
=&\, 0 - \left \langle \nabla_{E_1}E_2-c_3E_3, E_2\right \rangle\text{ since } [E_2,E_1]=-c_3E_3\\
=&\, 0-\left \langle \nabla_{E_1}E_2, E_2\right \rangle +c_3 \left \langle E_3, E_2\right \rangle\\
=&\, 0-\frac12 E_1 \left ( \left \langle E_2, E_2\right \rangle\ \right ) +0,\\
=&\, 0.
\end{split}
\end{equation}
Hence the $E_1$-component of $\nabla_{E_i}E_i$ (not summed on $i$) vanishes, for any $i$. But the same argument applies to all the other components, so the diagonal components of $\nabla_{E_i}E_j$ are zero (i.e., the Milnor frame is a \emph{geodesic basis} \cite{CLNN}). Then by \eqref{eq3.2} the diagonal of $B$ contributes zero to $\divergence B$.

Next we compute
\begin{equation}
\label{eq3.7}
\begin{split}
\left \langle E_1, \nabla_{E_2}E_3 + \nabla_{E_3}E_2\right \rangle = &\, -\left \langle \nabla_{E_2} E_1, E_3 \right \rangle-\left \langle \nabla_{E_3} E_1, E_2 \right \rangle\\
= &\, -\left \langle \nabla_{E_1} E_2, E_3 \right \rangle +c_3 -\left \langle \nabla_{E_1} E_3, E_2 \right \rangle -c_2,
\end{split}
\end{equation}
using properties of the Milnor frame. But by the Leibniz rule and the orthogonality of the basis, then
\begin{equation}
\label{eq3.8}
\left \langle \nabla_{E_1} E_2, E_3 \right \rangle = -\left \langle \nabla_{E_1} E_3, E_2 \right \rangle,
\end{equation}
so \eqref{eq3.7} simplifies to
\begin{equation}
\label{eq3.9}
\left \langle E_1, \nabla_{E_2}E_3 + \nabla_{E_3}E_2\right \rangle =c_3-c_2.
\end{equation}
It follows that, if $B$ is symmetric, then
\begin{equation}
\label{eq3.10}
B^{23}\nabla_{E_2}E_3+B^{32}\nabla_{E_3}E_2 = B^{23}\left ( \nabla_{E_2}E_3 + \nabla_{E_3}E_2 \right ) = \left ( c_3 -c_2\right ) B^{23},
\end{equation}
and likewise for the other off-diagonal components. Putting everything together, we recover \eqref{eq3.3}.
\end{proof}

There is a partial converse for Lemma \ref{lemma3.1}.

\begin{lemma}\label{lemma3.2}
Say that $(G,\gamma,X)$ is a solution of \eqref{eq1.1} where $g$ is a left-invariant metric on a unimodular $3$-dimensional Lie group $G$ and $X$ is Killing. Then $\kappa \left \langle X, \divergence B\right \rangle =0$.
\end{lemma}

\begin{proof}
Since $X$ is Killing, equation \eqref{eq1.1} becomes $\ric-\frac{1}{m}X\otimes X = \Lambda \gamma +\kappa B$. The Milnor frame for the Lie algebra of $G$ diagonalizes $\ric$ \cite[Theorem 4.3]{Milnor}, so we have that $-\frac{1}{m}X_iX_j=\kappa B_{ij}$ for $i\neq j$. If $\kappa= 0$, we are done, so assume otherwise. Then we may substitute $\kappa B^{12}=-\frac{1}{m}X^1X^2$, etc, into the right-hand side of \eqref{eq3.3}, which yields the result.
\end{proof}

\subsection{Proof of Theorem \ref{theorem1.2}}
Choose an orthonormal Milnor frame. In this basis, the Ricci tensor is $\ric=\diag \left ( \rho_1,\rho_2,\rho_3 \right )$. Define $B$ by requiring that in this frame it takes the form $B= \diag \left ( b_1,b_2,b_3 \right )$. Then by equation \eqref{eq3.3} we have that $\left \langle X, \divergence B \right \rangle=0$ for any choice of left-invariant $X$, so by Proposition \ref{proposition2.7} we have that $X$ is a Killing vector field. Furthermore, equation \eqref{eq2.25} will then apply, so $X$ will have at most one nonzero component in this basis. Without loss of generality, let $X=\diag (0,0,k)$, $k\neq 0$. Then equation \eqref{eq1.1} reduces to the matrix equation
\begin{equation}
\label{eq3.11}
\left [ \begin{array}{ccc} \rho_1 & 0 & 0 \\ 0 & \rho_2 & 0 \\ 0 & 0 & \rho_3 \end{array} \right ] -\frac{1}{m}\left [ \begin{array}{ccc} 0 & 0 & 0 \\ 0 & 0 & 0 \\
0 & 0 & k^2 \end{array} \right ] = \left [ \begin{array}{ccc}\Lambda & 0 & 0 \\ 0 & \Lambda & 0 \\
0 & 0 & \Lambda \end{array} \right ] + \kappa \left [ \begin{array}{ccc} b_1 & 0 & 0 \\ 0 & b_2 & 0 \\
0 & 0 & b_3 \end{array} \right ].
\end{equation}
Solutions with $\kappa\neq 0$ are given by
\begin{equation}
\label{eq3.12}
b_1=\frac{\rho_1-\Lambda}{\kappa},\quad b_2=\frac{\rho_2-\Lambda}{\kappa},\quad b_3=\frac{\rho_3-\Lambda-\frac{k^2}{m}}{\kappa}.
\end{equation}
If $b_1=b_2$, then $\rho_1=\rho_2$ and since the Ricci endomorphism has a repeated eigenvalue the conclusions of Theorem \ref{theorem1.1} apply.

This completes the proof of Theorem \ref{theorem1.2}.

For each choice of the $\rho_i$, there is a choice of the $b_i$ solving \eqref{eq1.1}. Now by equation \eqref{eq2.2}, if either $\beta=1$ or $\varepsilon=1$ or $\varepsilon=\beta$, then there are at most two distinct $\rho_i$, so if all the $\rho_i$ are distinct then the metric cannot be a Berger sphere. By \cite[Theorem 3.8]{HL} (or \cite[Theorem 2.10]{Shin}) the maximal isometry group is then $3$-dimensional.

There are solutions in which the Ricci endomorphism has a repeated eigenvalue, say $\rho_1=\rho_2$, and the metric is a Berger sphere, but the $b_i$ are all distinct. Even when all three eigenvalues $\rho_i$ are equal and the sphere is round, only two of the $b_i$ will be equal.

Finally, as $\kappa\to 0$, $B$ will remain well-defined if the numerators also approach zero at the same rate (or faster). But then $\rho_1\to\rho_2$ since they both must approach $\Lambda$, so $\ric$ acquires a repeated eigenvalue and the curvatures become those of a Berger sphere in this limit.

\appendix
\section{Myers-Perry-(A)dS spacetimes}
\setcounter{equation}{0}

\noindent Consider the family of stationary, 5-dimensional Myers-Perry-(A)dS$_5$ spacetimes $(\mathbf{M}, \mathbf{g})$ with equal angular momenta in two orthogonal planes \cite{GLPP}. The solutions are parameterized by their mass and angular momentum. These spacetimes satisfy the Einstein equations $\ric(\mathbf{g}) = \Lambda \mathbf{g}$, with  a group of isometries isomorphic to $\mathbb{R} \times \SU(2) \times \U(1)$. In Gaussian null coordinates, we have
\begin{equation}
\label{eqA.1}
\begin{aligned}
\mathbf{g}  =& \left[-\frac{(\rho + r_+)^2}{h(\rho)^2 g(\rho)^2} + \rho h(\rho)^2 j(\rho)^2 \right] \rho \td v^2 + 2 \td v \td r  +  2\rho h(\rho)^2  j(\rho)\td v \sigma_3  \\ &+ h(\rho)^2 \sigma_3^2 + (\rho + r_+)^2  (\sigma_1^2 + \sigma_2^2)
\end{aligned}
\end{equation}
where $r_+ >0$ is a parameter that characterizes the horizon scale. Explicit expressions for $g(\rho), h(\rho), j(\rho)$ can be read off from \cite[Section 2]{KLR} after a translation of the radial coordinate and the replacement $\ell^2 \to -4/\Lambda$ to allow for $\Lambda \geq 0$ in addition to $\Lambda <0$. The metric functions depend on the mass and angular momentum parameters $(\mu,a)$ which are related to $r_+$ via the constraint
\begin{equation}
\label{eqA.2}
\mu = \frac{r_+^4(4 - \Lambda r_+^2)}{2(4r_+^2 -a^2(4 - \Lambda r_+^2))} > 0
\end{equation}
where this inequality is required to avoid a naked singularity. The function $\rho = \rho(r)$ is defined implicitly by
\begin{equation}
\label{eqA.3}
r = \int_0^\rho \frac{\rho' + r_+}{h(\rho')} \; \td \rho'.
\end{equation}
The spacetimes are smooth on and outside the event horizon (which is also a Killing horizon) located along the null hypersurface $r = 0$ with null generator $\partial_v$. They are asymptotically flat, asymptotically anti-de Sitter, or asymptotically de Sitter for $\Lambda$ vanishing, negative, or positive, respectively. After fixing one parameter to match our scaling in \eqref{eq2.1},  the geometry near $\rho =0$ yields the following one-parameter family of solutions of \eqref{eq1.1} (with $m=2$):
\begin{equation}
\label{eqA.4}
\begin{split}
\gamma  =&\,  r_+^2\left[\left ( \sigma^1\right )^2 +  \left ( \sigma^2\right )^2\right] + \left ( \sigma^3\right )^2 ,\\
\kappa =&\, \frac{(4 - \Lambda r_+^2)r_+^2 -2}{2r_+^2},\qquad B = 2\left[\left ( \sigma^1\right )^2 +  \left ( \sigma^2\right )^2\right] + \frac{2}{r_+^2}(2r_+^2-1)(\sigma^3)^2,  \\
X = &\, \frac{2}{r_+} \left[(4 - \Lambda r_+^2)(1-r_+^2)\right]^{1/2}\sigma^3,
\end{split}
\end{equation}
with $r_* \leq r_+  \leq 1$ where $r_*$ is the smallest positive root of $\kappa$. In particular, $r_* = \frac{1}{\sqrt{2}}$ when $\Lambda =0$. Here $\Lambda \leq 2$, which ensures in the $\Lambda>0$ case that $\kappa\ge 0$, which in turn yields $4 - \Lambda r_+^2\ge \frac{2}{r_+^2}>0$ so $X$ is well defined in this case.

The extreme case is obtained by setting $\kappa=0$ to obtain
\begin{equation}
\label{eqA.5}\Lambda=\frac{4}{r_+^2}-\frac{2}{r_+^4}.
\end{equation}
Using this, we obtain the near horizon geometry
\begin{equation}
\label{eqA.6}
\begin{split}
\gamma=&\,r_+^2\left[\left ( \sigma^1\right )^2 +  \left ( \sigma^2\right )^2\right] + \left ( \sigma^3\right )^2 ,\\
X=&\, \pm \frac{2\sqrt{2\left ( 1-r_+^2\right )}}{r_+^2}\sigma^3.
\end{split}
\end{equation}

\end{document}